\documentclass[11pt]{article}
\usepackage[hmargin=1in,vmargin=1in]{geometry}
\usepackage{hchang}
\nolinenumbers

\usepackage{thm-restate}

\usepackage{comment}
\usepackage{tcolorbox}

\newtheorem{theorem}{Theorem}[section]
\newtheorem{lemma}[theorem]{Lemma}
\newtheorem{claim}[theorem]{Claim}
\newtheorem{observation}[theorem]{Observation}

\newtheorem{remark}[theorem]{Remark}
\newtheorem{corollary}[theorem]{Corollary}

\newtheorem{definition}[theorem]{Definition}
\newtheorem{question}[theorem]{Question}
\newtheorem{invariant}[theorem]{Invariant}

\newcommand{\cX}{\mathcal{X}}
\newcommand{\cN}{\mathcal{N}}
\newcommand{\dist}{\ensuremath{\delta}}
\newcommand{\ddim}{\ensuremath{d}}

\def\eps{\varepsilon}

\newcommand{\DS}{\texttt{DS}}
\newcommand{\points}{{\ensuremath{\DS.X}}}
\newcommand{\sparse}{{\ensuremath{\DS.S}}}
\newcommand{\net}{{\ensuremath{\DS.\cN}}}
\newcommand{\light}{{\ensuremath{\DS.L}}}
\newcommand{\len}{\ensuremath{\operatorname{len}}}

\newcommand{\lighti}[1]{\ensuremath{L[i]}}

\title{Dynamic Light Spanners in Doubling Metrics}
\date{}

\author{%
Sujoy Bhore%
\thanks{Department of Computer Science \& Engineering, Indian Institute of Technology Bombay. Email: \texttt{sujoy@cse.iitb.ac.in}. Work supported in part by ANRF ARG-MATRICS, Grant 002465.
}
\and
Jonathan Conroy%
\thanks{Department of Computer Science, Dartmouth College. Email: {\tt jonathan.conroy.gr@dartmouth.edu}. Supported by the U.S.\ National Science Foundation Grant No.\ CCF-2443017.
}
\and
Arnold Filtser%
\thanks{Department of Computer Science, Bar-Ilan University. Email: \texttt{arnold.filtser@biu.ac.il}. This research was supported by the ISRAEL SCIENCE FOUNDATION (grant No.\ 1042/22).
}}

\begin{document}

\maketitle

\begin{abstract}
A $t$-spanner of a point set $X$ in a metric space $(\mathcal{X}, \delta)$ is a graph $G$ with vertex set $P$ such that, for any pair of points $u,v \in X$, the distance between $u$ and $v$ in $G$ is at most $t$ times $\delta(u,v)$. We study the problem of maintaining a spanner for a dynamic point set $X$---that is, when $X$ undergoes a sequence of insertions and deletions---in a metric space of constant doubling dimension. For any constant $\varepsilon>0$, we maintain a $(1+\varepsilon)$-spanner of $P$ whose total weight remains within a constant factor of the weight of the minimum spanning tree of $X$. Each update (insertion or deletion) can be performed in $\operatorname{poly}(\log \Phi)$ time, where $\Phi$ denotes the aspect ratio of $X$. Prior to our work, no efficient dynamic algorithm for maintaining a light-weight spanner was known even for point sets in low-dimensional Euclidean space.
\end{abstract}

\section{Introduction}
Let $X\subset \cX$ be a set of points in the metric $(\cX,\delta_X)$. A graph $G= (X,E,\delta)$ with vertex set $X$ is a \EMPH{geometric graph} if the weight function on the edges is the distance function $\delta$. In other words, the weight of each edge is the distance between its endpoints in $\cX$. Observe that if $G$ is a geometric graph, then by the triangle inequality, for every $x,y\in X$, $\delta_G(x,y)\geq \delta(x,y)$.  Given a parameter $t$, called the \EMPH{stretch}, we say that  $G$ is a (geometric) \EMPH{$t$-spanner} for $X$ if $\dist_G(x,y)\leq t\cdot \delta(x,y)$ for every $x,y\in X$. The natural goal is to construct, for a given stretch $t$, a $t$-spanner with small number of edges and a small total weight.

\EMPH{Lightness} and \EMPH{sparsity} are two central measures for evaluating spanners. For a spanner $G=(X,E)$, the lightness is the ratio $w(G)/w(\mathrm{MST})$ between the total weight of $G$ and the weight of a minimum spanning tree on $X$. The sparsity of $G$ is the ratio $|E(G)|/|E(\mathrm{MST})| \approx |E(G)|/|X|$, comparing the number of edges of $G$ to that of an MST. 
Since every spanner (with finite stretch) is connected and therefore contains a spanning tree, the lightness and sparsity of $G$ 
constitute a measure of how close a spanner (specifically $w(G)$ and $|E(G)|$) to the optimal weight and the optimal number of edges of a spanner with finite stretch.
A long line of research has studied spanners in Euclidean spaces~\cite{arya1995euclidean, rao1998approximating,HarPeledIS13,
 ElkinS15,FiltserN22,BhoreT22,LeS25}, doubling~\cite{gao2004deformable,chan2009small,smid2009weak, gottlieb2015light, FS20, BLW19, BhoreM25}, planar and minor-free metric~\cite{borradaile2017minor, grigni2002light, cohen2020light, BhoreKK0LPT25}, and general graphs~\cite{peleg1989graph,althofer1993sparse,ElkinS16,ElkinNS15,chechik2018near,FS20,LeS23,Bodwin24}. 
Of particular interest, we focus on \emph{low-dimensional Euclidean space}, and more generally on \emph{doubling metrics}.

\paragraph{Doubling metrics.} Let $B(x,r) = \{y | \delta(x,y)\leq r\}$ be a ball of radius $r\geq 0$ centered at a point $x\in \cX$. A metric $(\cX,\delta)$ has \EMPH{doubling dimension $\ddim$} if for every $r > 0$, every ball of radius $r$ can be covered by at most $2^{\ddim}$ balls of radius $r/2$: $\forall x\in X, r > 0,~\exists Y\subseteq X: |Y|\leq 2^{\ddim} \text{ and } B_X(x,r) \subseteq \cup_{y\in Y} B_X(y,r/2)$. 
The study of spanners in doubling metrics was initiated by Gao, Guibas, and Nguyen~\cite{gao2004deformable}, who proved that any $n$-point set in a metric space of doubling dimension $\ddim$ admits a $(1+\varepsilon)$-spanner with $\varepsilon^{-O(\ddim)} n$ edges.
Smid~\cite{smid2009weak}, via an analysis of the \EMPH{greedy algorithm}, showed that greedy $(1+\varepsilon)$-spanners have $\varepsilon^{-O(\ddim)} n$ edges and achieve lightness $\varepsilon^{-O(\ddim)} \log n$. While it has long been known that $d$-dimensional Euclidean point sets admit $(1+\varepsilon)$-spanners with lightness $\varepsilon^{-O(d)}$, it remained a major open problem whether an analogous result holds for metric spaces of low doubling dimension. 
In a major breakthrough, Gottlieb~\cite{gottlieb2015light} proved that any metric space of doubling dimension $\ddim$ admits a $(1+\varepsilon)$-spanner with lightness $(\ddim/\varepsilon)^{O(\ddim)}$. 
Gottlieb devised a complicated spanner construction instead of the classic greedy spanner.
Later, Filtser and Solomon \cite{FS20} showed that the greedy spanner is existentially optimal, concluding that it has lightness $(\ddim/\varepsilon)^{O(\ddim)}$ as well.
Finally, Borradaile, Le, and Wulff-Nilsen \cite{BLW19} showed the greedy spanner has
weight $\varepsilon^{-O(\ddim)}$, which is optimal.

\paragraph{Dynamic Spanners.} The problem of maintaining a sparse $(1+\e)$-spanner for a point set $X$ under insertions and deletions has also been studied over the years. Already back in 1999, Arya et al.~\cite{arya1999dynamic} obtained polylogarithmic update time in a restricted model with random updates
(a deletion removes a uniformly random point from $S$, and an insertion results in a uniformly random point in the updated set).
Later, Bose et al.~\cite{bose2004ordered} studied the problem in the context of $\theta$-graphs, and presented an algorithm that supports insertions only and achieves polylogarithmic update time. 
Gao et al.~\cite{gao2004deformable} studied kinetic and dynamic geometric spanners for points in bounded doubling dimensions, and obtained a fully dynamic spanner 
with worst-case update time $O(\log \Phi)$, where $\Phi$ is the 
\EMPH{aspect ratio} of the point set $X$, defined as the ratio between the maximum and minimum pairwise distances in $X$. Krauthgamer and Lee~\cite{KL04} obtained similar results for proximity search, where the time bounds of their dynamic updates depend also in the aspect ratio. For proximity search, later, Cole and Gottlieb~\cite{cole2006searching} removed the dependency on the aspect ratio. For the Euclidean plane, Abam et al.~\cite{abam2008simple} presented an algorithm with $O(\log n)$ update time. 
For doubling metrics, Roditty~\cite{roditty2012fully} presented the first fully dynamic algorithm for maintaining a geometric $(1+\varepsilon)$-spanner, with update time depending only on the number of points in $X$. The algorithm obtained in~\cite{roditty2012fully} supports insertions in amortized $O(\log n)$ time and deletions in amortized $\tilde{O}(n^{1/3})$ time
where $\tilde{O}$ hides polylogarithmic factors, 
and maintains spanner with $O(n / \varepsilon^{\ddim})$ edges.
Gotllieb and Roditty~\cite{GottliebR08}
gave a dynamic spanner that supports insertions in \(O\!\left(\varepsilon^{-O(\ddim)} \log^{2} n\right)\) amortized time and deletions in \(O\!\left(\varepsilon^{-O(\ddim)} \log^{3} n\right)\) amortized time.
Later, in~\cite{gottlieb2008optimal},
\footnote{Gotllieb and Roditty \cite{gottlieb2008optimal} additionally assume that they can query the distance between a deleted point and a non-deleted point, which (for example) is possible for Euclidean metrics.}
they gave a \((1+\varepsilon)\)-spanner of constant degree that supports updates in $\varepsilon^{-O(\ddim)} \log n$ worst-case time.

Roughly speaking, these ``classic'' constructions of dynamic spanners for doubling metric are based on a \EMPH{net tree} (see Section~\ref{sec:NetTreesGreedySpanner}). 
While it is straightforward to construct a sparse spanner from a net tree, this tool is too rough to obtain lightness.
Indeed, in Figure~\ref{fig:NetTreeHeavy} we illustrate the net tree of the path graph, and show that it has lightness $\Omega(\log n)$.
    \begin{figure}[t]
	\begin{center}
		\includegraphics[width=0.8\textwidth]{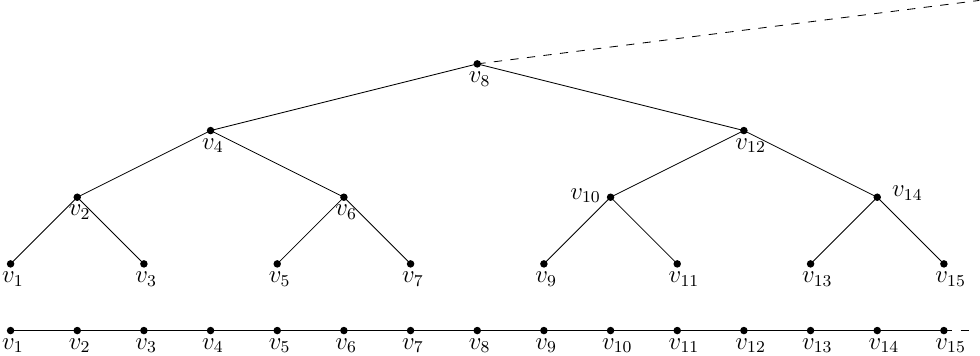}
		\caption{
        On bottom: The path graph $P_n$. All edges are unit weight. 
        On top: 
        a depiction of the net tree $T$ of $P_n$. There is a laminar hierarchy of nets; the $2^i$-net is $N_i=\{v_j\in P_n\mid j\mod 2^i=0\}$ (all vertices at hight $i$ in the tree and above).
        The net-tree spanner contains all the tree edges, and many additional edges.
        The weight of the edge $\{v_i,v_j\}$ is $|i-j|$.
        The tree $T$ has depth $\log n$, and the total weight of the edges going from depth $i$ vertices to depth $i+1$ is $\Theta(n)$.
        Thus $w(T)=\Theta(n\cdot\log n)$.
        As the weight of the MST of $P_n$ is $n-1$, the net-tree spanner has lightness $\Omega(\log n)$.
		}\label{fig:NetTreeHeavy}
	\end{center}
\end{figure}
A more recent approach to constructing dynamic sparse spanners is using \EMPH{locality-sensitive orderings (LSO)} (introduced by  Chan, Har-Peled, and Jones~\cite{chan2020locality}). In Section~\ref{sec:LSO}, we summarize the history of LSO, and prove that LSO-based spanners have logarithmic lightness (even for $\R^1$).

While there is a substantial body of work on dynamic spanners in Euclidean metrics and, more generally, in doubling metrics, as summarized above, there is comparatively little work on dynamic \emph{light} spanners. Controlling lightness is arguably a significantly harder objective, as evidenced by the extensive literature in the offline setting (see, e.g.,~\cite{LeS25, arya1995euclidean, rao1998approximating, borradaile2017minor, LeS23, NS07, BhoreM25, FS20, gottlieb2015light, ElkinS15, BhoreT22,CCLST25} and the references therein).

\begin{question}
\label{q:main}
Can we design a dynamic algorithm that maintains a \EMPH{light-weight} spanner for a point set under dynamic updates in doubling metrics with polylogarithmic update time?
\end{question}

Partial progress toward Question~\ref{q:main} was made by Eppstein and Khodabandeh~\cite{EK24}, who showed that one can maintain a \EMPH{low-recourse} light spanner for dynamic point sets in low-dimensional Euclidean space. The \emph{recourse} of a dynamic spanner is the number of edges that change after each insertion or deletion. They obtained a construction with amortized $O(1)$ recourse per insertion and amortized $O(\log \Phi)$ recourse per deletion. However, their result addresses recourse rather than efficiency: the actual time to update the spanner can be very large (no faster than just recomputing a spanner from scratch), and therefore their work does not resolve the question of achieving polylogarithmic (or even sublinear) update time.

\subsection{Our Contribution}

We resolve Question~\ref{q:main} in the affirmative. We say a set of points $X$ is $(a,b)$-bounded if $\min_{u,v \in X} \dist(u,v) \ge a$ and $\max_{u,v \in X} \dist(u,v) \le b$.
\begin{theorem}
\label{thm:main}
    Let $(\cX, \dist)$ be a metric space with doubling dimension $\ddim$, and let $\e >0$ and $\Phi > 0$ be parameters. Let $X \subseteq \cX$ be a set of points undergoing insertions and deletions, such that at all times $X$ is $(1, \Phi)$-bounded. We maintain a $(1+\e)$-spanner $L$ of $X$ with lightness at most $\e^{-O(\ddim)}$, where each insertion/deletion to $X$ can be processed in time $(\e^{-1} \cdot \log \Phi)^{O(\ddim)}$.
\end{theorem}
For the sake of simplicity, we state our theorem as though the maximum distance $\Phi$ must be given in advance, and the minimum distance between points in $X$ must be at least 1. Using standard techniques \cite{KL04,EK24} we could remove these restrictions at no loss in the running time: we obtain a data structure whose update time depends adaptively on the aspect ratio $\Phi$ of the current point set $X$. See Remark~\ref{rem:aspect-ratio}. 
Additionally, we remark that we do not need to know the doubling dimension $\ddim$ in advance; the value $\ddim$ only appears in the runtime analysis.

\paragraph{Techniques.} Our proof begin somewhat similarly to Eppstein and Khodabandeh \cite{EK24}: roughly speaking, we seek to maintain a \emph{greedy spanner} on top of a \emph{dynamic sparse spanner} of points $X$ (crucially, this spanner is much more structured than just the greedy spanner on the complete geometric graph on $X$). In Section~\ref{S:recourse}, we first give a new (and arguably much simpler\footnote{We avoid the \emph{bucketing} technique in the algorithm of \cite{EK24} and instead work with all spanner edges at once; we maintain spanner invariants that work in doubling metrics and not just Euclidean metrics (unlike the \emph{leapfrog property} used in \cite{EK24}); and finally, we give a simple worst-case bound on the recourse of $O(\log \Phi)$, rather than using a potential function to bound the amortized recourse.}) algorithm and analysis that such a spanner can be maintained with bounded recourse.
After each insertion/deletion, we identify $O(\log \Phi)$ edges $(u,v)$ of our greedy spanner $L$ which might have to change, to satisfy the greedy spanner invariants: roughly speaking, we examine the distance $\dist_{L<}(u,v)$ between $u$ and $v$ in the graph consisting of all of edges of $L$ that are shorter than $(u,v)$\footnote{Our description here assumes we want to maintain a greedy spanner, but actually we maintain a slightly more robust type of spanner which we call the \emph{delayed greedy spanner}: we work with distances $\dist^*(u,v)$ rather than $\dist_{L<}(u,v)$, as defined in Section~\ref{SS:invariants}. We do this to control cascading changes to $L$.}, and we insert or delete $(u,v)$ to $L$ based on whether $\dist_{L<}(u,v) \le (1+\e) \dist(u,v)$. Our key technical contribution  (in Section~\ref{S:distance-oracle}) is to show that we can efficiently maintain estimates of $\dist_{L<}(u,v)$ even as we insert or delete edges into $L$.

\subsection{Related Work} 
\label{SS:related-work}

\paragraph*{Dynamic Graph Spanners.}
In general graphs, the goal is to maintain, under edge updates, for any integer $k \ge 2$, a spanner with stretch $2k - 1$ and $\widetilde{O}(n^{1+1/k})$ edges. Assuming the Erd\H{o}s--Girth conjecture, this trade-off is essentially optimal. The dynamic spanner problem was initiated by Ausiello et al.~\cite{ausiello2005small}, and has since been studied extensively; see, for example,
\cite{baswana2012fully, bodwin2016fully, elkin2011streaming, bernstein2021deamortization, forster2019dynamic}.
However, these works focus primarily on \EMPH{sparsity}, and the resulting techniques do not naturally extend to maintaining \EMPH{light} spanners.
Moreover, due to inherent degree barriers in general graphs, this line of work exclusively considers \emph{edge updates}.
In contrast, in geometric and other structured metric spaces, vertex updates are often the more natural model, since the underlying graph need not be maintained explicitly.
This distinction is crucial for light spanners, whose construction and analysis rely fundamentally on geometric structure rather than purely combinatorial sparsification.

\paragraph*{Online Spanners.}
Spanners have also been studied extensively in the closely related \emph{online} model.
In this setting, points from a metric space arrive one by one, and at each step, the algorithm is given the subset of points that arrived so far.
The goal is to maintain a $t$-spanner at all times. Unlike in the fully dynamic setting, the algorithm is allowed to \emph{add} edges when a new point arrives, but may never remove previously added edges.
In addition, the total number of points $n$ is not known in advance. A number of algorithms and lower bounds have been obtained for this model, which addressed both sparsity and lightness guarantees along with stretch, under online constraints; see, e.g.,
\cite{BhoreFKT24, BhoreT25, AlonAABN06, AlonA92, BermanC97, HajiaghayiLP17, NaorPS11}.
These results achieve near-optimal trade-offs in the online regime and highlight the fundamental differences between incremental constructions and fully dynamic settings.

\section{Net trees and greedy spanners}\label{sec:NetTreesGreedySpanner}
\paragraph{Preliminaries.} Throughout this paper, let \EMPH{$M = (\cX, \dist)$} be a metric space with constant doubling dimension $\ddim = O(1)$.
For any vertex $x \in \cX$ and radius $r > 0$, we let \EMPH{$B(x, r)$} denote the set of points within distance $r$ of $x$; that is, $B(x,r) \coloneqq \set{y \in \cX : \dist(x,y) \le r}$.
If $G$ is a (geometric) graph, we use the notation \EMPH{$\dist_G(\cdot,\cdot)$} to denote shortest-path distances on $G$. For any positive integer $n$, we write \EMPH{$[n]$} to mean $\set{i \in \Z: 0 \le i \le n}$. Throughout this paper, we fix some sufficiently small positive integer \EMPH{$\e> 0$} and aim to construct a light $(1+O(\e))$-spanner; our Theorem~\ref{thm:main} will follow after rescaling $\e \gets \e / O(1)$.

\subsection{Net-Tree Spanners}

In this subsection, we review the classic net-tree spanner \cite{gao2004deformable}.

\paragraph{Net tree.}  Let $X \subseteq \cX$ be a set of points. Assume that the smallest distance between points in $X$ is $1$, and the largest distance is \EMPH{$\Phi$}; for simplicity we assume $\Phi = 2^i$ for some integer $i$.
A \EMPH{$\Delta$-net} of $X$ is a set of points $N$ such that every point $x \in X$ is within distance $\Delta$ of some point in $N$, and any two net points in $N$ are at distance greater than $\Delta$. A $\Delta$-net can be constructed greedily.
A \EMPH{net hierarchy} of $X$ is a hierarchical sequence of nets $(N_0, N_1, \ldots, N_{\log \Phi})$ such that every $N_i$ is a subset of $N_{i-1}$, and $N_i$ is a $2^{i}$-net of $N_{i-1}$;
we take $N_0 = X$. Observe that (by a geometric series) every point in $X$ is within distance $2^{i+1}$ of some net point in $N_i$. Moreover, the \emph{packing bound} of doubling metrics implies that there are not too many points close together in each $\Delta$-net.
\begin{observation}[Packing bound]
\label{obs:packing}
    Let $N$ be a set of points that are pairwise at distance at least $\Delta$. For any $\alpha \ge 1$ and $x \in \cX$, there are at most $\alpha^{O(\ddim)}$ points in $B(x,\alpha \cdot \Delta) \cap N$.
\end{observation}
An implicit representation of the net hierarchy $\cN$ can be maintained using a data structure called the \EMPH{net tree} \cite{gao2004deformable,KL04,HM06}.
We will sometimes abuse terminology and call $\cN$ a net tree.

\paragraph{Net-tree spanner.}
The net-tree spanner was introduced (under the name ``deformable spanner'') by Gao, Guibas, and Nguyen \cite{gao2004deformable}; similar ideas were independently designed by Krauthgamer and Lee \cite{KL04}.  The \EMPH{$\e$-net-tree spanner} of $X$ with respect to the net tree $\mathcal{N} = (N_0, \ldots, N_{\log \Phi})$ is the graph with vertex set $X$ and with the following edge set: for every scale $i \in [\log \Phi]$, add an edge between any two vertices $u, v \in N_i$ if $\dist(u, v) \le c \cdot 2^i$, where $c \coloneqq 4 + 16 \e^{-1}$.
\begin{lemma}[{\cite[Theorem 3.2]{gao2004deformable}}]
\label{lem:net-spanner}
    An $\e$-net-tree spanner of $X$ is a $(1+\e)$-spanner of $X$.
\end{lemma}
For any pair of vertices $(u, v)$, we say the \EMPH{scale} of $(u, v)$ is the value $i$ such that $\dist(u, v) \in [2^{i-1}, 2^i)$. Observe that an edge $(u,v)$ at scale $i$ in an $\e$-net-tree spanner has both endpoints in the $O(\e) \cdot 2^i$-net of the net tree
(not the $2^i$-net). The packing bound implies that every vertex in the $\e$-net-tree spanner has degree at most $\e^{-O(\ddim)} \cdot \log \Phi$.
In fact, using a charging argument,
one can show that the $\e$-net-tree spanner contains at most $\e^{-O(\ddim)} \cdot n$ edges: that is, it is \emph{sparse}.
We frequently denote the $\e$-net-tree spanner as \EMPH{$S$}.

\paragraph{Dynamic maintenance.}
Gao, Guibas, and Nguyen \cite{gao2004deformable} showed that the net tree spanner can be maintained dynamically. We summarize their guarantee below.
\begin{lemma}[Rephrasing of {\cite[Theorem 4.2]{gao2004deformable}}]
\label{lem:deformable}
    There is a data structure that implicitly maintains a hierarchy of nets $\mathcal{N}=(N_0, N_1, \ldots, N_{\log \Phi})$ of a point set $X$ undergoing insertions and deletions. Each insertion or deletion can be processed in $O(\log \Phi)$ time.
    Moreover\footnote{Note that the phrasing of \cite[Theorem 4.2]{gao2004deformable} does not explicitly separate the maintenance of the net hierarchy $\cN$ from the maintenance of the spanner (they just state the existence of a dynamic spanner); however, our phrasing of Lemma~\ref{lem:deformable} is immediate from their proof. In our application, we will need to dynamically maintain an $\e_1$-net-tree spanner of $X$ and an $\e_2$-net-tree spanner of $X$ (for two different values $\e_1 \neq \e_2$) \emph{with respect to the same net tree $\cN$}. This is possible using the data structure of \cite{gao2004deformable}.}, for any constant $\e > 0$, one can maintain an $\e$-net-tree spanner of $X$ with respect to $\mathcal{N}$, in $O(\log \Phi)$ time per insertion or deletion.
\end{lemma}
We will make use of Lemma~\ref{lem:deformable}, but we will have to slightly unwrap the black box --- we will need some control on how the net tree $\cN$ changes after each insertion or deletion. Fortunately, \cite{gao2004deformable} update the net tree $\mathcal{N}$ in a simple and structured way. Inserting a point $x$ into $X$ has the following result, described in Figure~\ref{fig:sparse-insert}.
\begin{figure}[h!]
    \centering
    \begin{algorithm}
    $i_0 \gets$ smallest scale such that the scale-$i_0$ net $N_{i_0} \in \cN$ satisfies $\dist(x, N_{i_0}) < 2^{i_0}$
    \\ update $\cN$ by inserting $x$ into net $N_i \in \cN$ for all $i < i_0$
    \end{algorithm}
    \caption{Pseudocode for \cite{gao2004deformable} algorithm to update net $\mathcal{N}$ after inserting point $x$ into $X$}
    \label{fig:sparse-insert}
\end{figure}

\noindent Deleting a point $x$ from $X$ is described in Figure~\ref{fig:sparse-delete}.
\begin{figure}[h!]
    \centering
    \begin{algorithm}
    update $\cN$ by deleting $x$ from every net $N_i \in \cN$
    \\ for every scale $i \gets 1, 2, \ldots, \log \Phi$:\+
    \\      while there exists some point $y \in N_{i-1}$ with $\dist(u, N_{i}) > 2^i$, add $y$ to $N_i$\-
    \end{algorithm}
    \caption{Pseudocode for \cite{gao2004deformable} algorithm to update net $\mathcal{N}$ after deleting point $x$ from $X$}
    \label{fig:sparse-delete}
\end{figure}

We will not discuss the details of the data structures that \cite{gao2004deformable} maintain to implement these two operations efficiently. We have unwrapped the \cite{gao2004deformable} black box so that we can prove that every insertion/deletion of a point $x$ only changes $\cN$ ``locally'' nearby $x$.
\begin{lemma}
\label{lem:net-delete-local}
Let $\cN^{\rm old} = (N_0^{\rm old}, \ldots, \cN_{\log \Phi}^{\rm old})$ be a net tree for point set $X$. Let $\cN = (N_0, \ldots, N_{\log \Phi})$ be the net tree produced from $\cN^{\rm old}$ by the \cite{gao2004deformable} algorithm after a point $x$ is inserted into (resp. deleted from) $X$. For any scale $i \in [\log \Phi]$, every point $y$ in the symmetric difference of $N_i$ and $N^{\rm old}_i$ satisfies $\dist(y,x) \le 2^i$.
\end{lemma}
\begin{proof}
When $x$ is inserted into $X$ (according to the pseudocode in Figure~\ref{fig:sparse-insert}), the lemma follows trivially from the \cite{gao2004deformable} algorithm description: the only point that could differ between $N_i$ and $N^{\rm old}_i$ is $x$. We now argue that the lemma holds when $x$ is deleted from $X$ (according to the pseudocode in Figure~\ref{fig:sparse-delete}).
    By construction, the only point that could be removed from $N_i^{\rm old}$ is $x$; that is, the set $N_i^{\rm old} \setminus N_i$ is either $\varnothing$ or $\set{x}$. So it remains to show that every point $y$ in $N_i \setminus N^{\rm old}_i$ satisfies $\dist(y,x) \le 2^i$.
    The proof is by induction on $i$. When $i=0$, every point in $\points$ is contained in $N^{\rm old}_0$, so $N_0 \setminus N^{\rm old}_0 = \varnothing$ and the claim holds trivially. Now consider $i > 0$. Consider some point $y \in N_i \setminus N_i^{\rm old}$. By the algorithm description, the point $y$ satisfies
    \[\dist(y, N^{\rm old}_{i} \setminus \set{x}) > 2^i.\]
    There are two cases. In the first case, suppose $y \in N^{\rm old}_{i-1}$. Because $N^{\rm old}_{i}$ is a $2^i$-net for $N^{\rm old}_{i-1}$, we have $\dist(y, N^{\rm old}_i) \le 2^i$. We conclude that $\dist(y, x) \le 2^i$ as desired.
    In the second case, suppose $y \not \in N^{\rm old}_{i-1}$. Thus, $y \in N_{i-1} \setminus N^{\rm old}_{i-1}$, and by induction we have that $\dist(y, x) \le 2^{i-1} < 2^i$ as desired.
\end{proof}

Finally, we need one more guarantee from the \cite{gao2004deformable} data structure, which appears (implicitly) in the proof of \cite[Theorem 4.2]{gao2004deformable}.
\begin{claim}[\cite{gao2004deformable}]
\label{clm:fast-ball}
    Consider the data structure of \cite{gao2004deformable} that maintains a hierarchy of nets $(N_0, N_1, \ldots, N_{\log \Phi})$ of $X$  and an $\e$-net-tree spanner. Let $c_\e = 4 + 16\e^{-1}$.
    Given a scale $i \in [\log \Phi]$ and a query point $x \in \mathcal X$,  one can find all vertices in $N_i \cap B(x, c_\e \cdot 2^i)$ in time $\e^{-O(\ddim)} = O(1)$.
\end{claim}
In our applications, we will want something a bit more general: we fix a sufficiently large constant $\alpha$ (say, $\alpha > 100$ suffices) and we want to find all vertices in $N_i \cap B(x, \alpha \cdot c_\e \cdot 2^i)$ in time $\e^{-O(\ddim)}$. Observe that Claim~\ref{clm:fast-ball} allows us to do this: indeed, we could simply maintain the data structure for an $\e/8$-net-tree spanner, which (by Claim~\ref{clm:fast-ball}) allows us to quickly find all vertices in $N_i \cap B(x, c_{\e/8} \cdot 2^i) \supseteq N_i \cap B(x, 100 \cdot c_\e \cdot 2^i)$.
\begin{corollary}
\label{cor:expanded-fast-ball}
    Consider the data structure of \cite{gao2004deformable} that maintains a hierarchy of nets $(N_0, N_1, \ldots, N_{\log \Phi})$ of $X$  and an $\e$-net-tree spanner. Let $c_\e = 4 + 16\e^{-1}$.
    Given a scale $i \in [\log \Phi]$ and a query point $x \in \mathcal X$,  one can find all vertices in $N_i \cap B(x, 10 \cdot c_\e \cdot 2^i)$ in time $\e^{-O(\ddim)} = O(1)$.

\end{corollary}

\subsection{Invariants for Our Light Spanner: Delayed Greedy Spanner}
\label{SS:invariants}

While the net-tree spanner $S$ is sparse and can be maintained dynamically, it could have lightness $\Omega(\log n)$ (see Figure~\ref{fig:NetTreeHeavy}).
We aim to find a subgraph of $S$ with $O(1)$ lightness. 
Intuitively, we want to maintain a greedy $(1+\e)$-spanner $H$ of the graph $S$: the greedy spanner (and even the \emph{approximate} greedy spanner) are light \cite{BLW19,FS20}. Recall that the greedy spanner initializes $H \gets \varnothing$ and processes edges in $S$ from smallest to largest; when we process $(u,v)$, we add $(u,v)$ to $H$ iff $\dist_{H}(u,v) > (1+\e) \dist_S(u,v)$. In other words, the greedy spanner satisfies: for every edge $(u,v)$ in $S$, letting $H_{<(u,v)}$
denote the set of edges in $H$ with smaller\footnote{assume all edge lengths in $S$ are distinct, by breaking ties arbitrarily} lengths than $(u,v)$, we have 
$\dist_{H_{<(u,v)}}(u,v) > (1+\e) \dist_S(u,v)$ iff $(u,v) \in L$.

\paragraph{Delayed greedy spanner.} We maintain a light spanner, denoted $\EMPH{$S$} \subseteq S$
that doesn't obey the greedy spanner invariants exactly; rather, we define a more robust set of invariants. 
For any scale $i \in [\log \Phi]$, we define the set of edges $\EMPH{$\lighti i$}\subseteq L$ to be the set of all edges in $L$ with \emph{scale strictly less than} $i$; that is, every edge in $\lighti i$ has length at most $2^{i-1}$. 
For any edge $(u,v)$ of the $\e$-net-tree spanner at scale $i$, we define \EMPH{$\dist^*(u,v)$}$\coloneqq \dist_{\lighti i}(u,v)$; recall that $\dist_{\lighti i}(\cdot, \cdot)$ denotes the shortest-path metric in the geometric graph induced by $\lighti i$.
To construct our light spanner, we maintain a subset of edges $L \subseteq S$ with the following properties:
\begin{invariant}
\label{inv:stretch}
    Every edge $(u,v) \in S \setminus L$ satisfies $\dist^*(u,v) \le (1+\e) \cdot \dist(u,v)$.
\end{invariant}
\begin{invariant}
\label{inv:light}
    Every edge $(u,v) \in L$ satisfies $\dist^*(u,v) > (1+\frac \e 3) \cdot \dist(u,v)$.
\end{invariant}

If $L$ satisfies Invariants~\ref{inv:stretch} and \ref{inv:light}, we call $L$ a \EMPH{delayed greedy spanner} of $S$
(because there is a delay between the time an edge is added to the spanner and the time it is taken into account when making the next decision). 
We emphasize that the delayed greedy spanner invariants are different than the greedy spanner invariants for $H$: in our invariants, the choice of adding $(u,v)$ to the spanner depends on $L[i]$ (the set of edges in $L$ with \emph{scale} strictly smaller than the scale of $(u,v)$), rather than $H_{<(u,v)}$ (the set of edges in the greedy spanner $H$ that are strictly shorter than $\dist(u,v)$). Nevertheless, we show that any set of edges $L$ that satisfies Invariants~\ref{inv:stretch} and \ref{inv:light} is a $(1+O(\e))$-spanner (Lemma~\ref{lem:inv-implies-stretch}) and has constant lightness (Lemma~\ref{lem:inv-implies-light}). 
\begin{restatable}{lemma}{invImpliesStretch}
    \label{lem:inv-implies-stretch}
Let $S$ be an $\e$-net-tree spanner for $X$, and let $L \subseteq S$. 
If Invariant~\ref{inv:stretch} holds, then $L$ is a $(1+3 \e)$-approximate spanner for $(X, \dist)$.
\end{restatable}
\begin{proof}
    By Lemma~\ref{lem:net-spanner}, the $\e$-net-tree spanner $S$ is a $(1+\e)$-spanner for $(X, \dist)$. Moreover, Invariant~\ref{inv:stretch} implies that $L$ is a $(1+\e)$-spanner for $S$: by triangle inequality it suffices to check that every edge of $S$ is preserved up to $1+\e$ approximation in $L$, and indeed, for every edge $(u,v) \in S$, we have $\dist_L(u,v) \le \dist^*(u,v) \le (1+\e) \dist(u,v) \le (1+\e) \dist_S(u,v)$. We conclude that $L$ is a spanner for $(X,\delta)$ with stretch $(1+\e)^2 \le (1+3\e)$.
\end{proof}

\begin{restatable}{lemma}{invImpliesLight}
\label{lem:inv-implies-light}
Let $S$ be an $\e$-net-tree spanner for $X$, and let $L \subseteq S$. If Invariants~\ref{inv:stretch} and \ref{inv:light} hold, then $L$ has lightness $\e^{-O(\ddim)}$.    
\end{restatable}

\begin{proof}
    We use, as a black box, the fact that the greedy $(1+\frac \e 3)$-spanner in a metric space of doubling dimension $d$ has lightness $\e^{-O(d)}$ \cite{BLW19}. We consider the shortest-path metric $\dist_L$ induced by the edges of $L$; because $(X, \dist)$ has doubling dimension $\ddim$, the stretch bound of Lemma~\ref{lem:inv-implies-stretch} together with the definition of doubling dimension implies that the doubling dimension of $(X, \dist_L)$ is $O(\ddim)$ (see \cite[Observation 7]{FS20}).
    Consider building the greedy $(1+ \frac \e 3)$-spanner \EMPH{$\hat L$} of $L$: that is, we initialize $\hat L \gets \varnothing$, then iterate over the edges of $L$ from smallest to largest\footnote{by triangle inequality, it suffices to examine only the edges of $L$ when building the greedy spanner, rather than all pairs of vertices in $L$}, and add an edge $(u,v) \in L$ into the spanner $\hat L$ iff $\dist_{\hat L}(u,v) > (1+ \frac \e 3) \dist_L(u,v)$. 
    The greedy spanner $\hat L$ has lightness $\e^{-O(\ddim)}$.
    We now \EMPH{charge} every edge of $L$ to an edge of $\hat L$. If an edge $(u,v) \in L$ is added to $\hat L$, then charge $(u,v)$ to itself. 
    Otherwise, suppose $(u,v) \in L$ is not added to $\hat L$. This means that there is a path $P$ between $u$ and $v$ in $\hat L$ (at the instant that $(u,v)$ is considered) with length at most $(1+ \frac \e 3) \dist_L(u,v)$. 
    Because $S$ is a geometric graph, $\dist_S(u,v) = \dist(u,v)$; i.e., $P$ has length at most $(1+ \frac \e 3)\dist(u,v)$. Let $i$ be the scale of $(u,v)$, meaning that $\dist(u,v) \in [2^{i-1}, 2^i)$. 
    The path $P$ contains some scale-$i$ edge $(u',v')$ in the greedy spanner $\hat L$: otherwise, if $P$ contained only scale-$j$ edges with $j < i$, we would have $\dist^*(u,v) \le (1+ \frac \e 3) \cdot \dist(u,v)$, contradicting Invariant~\ref{inv:light}. Charge $(u,v)$ to the scale-$i$ edge $(u',v')$.
    
    We now claim that every edge $(u',v')$ in $\hat L$ is charged at most $\e^{-O(\ddim)}$ times. Indeed, whenever a scale-$i$ edge $(u',v')$ is charged by a scale-$i$ edge $(u,v)$, we have $\dist(u,u') \le (1+\frac \e 3) 2^i$; this is because, by definition of charging, $u'$ lies on a path $u$ and $v$ with length at most $(1+ \frac \e 3) \dist(u,v) \le (1+ \frac \e 3) 2^i$. The packing bound (Observation~\ref{obs:packing}) implies that there are only $\e^{-O(\ddim)}$ scale-$i$ edges that could charge $(u',v')$. Because we only charge scale-$i$ edges in $L$ to scale-$i$ edges in $\hat L$, we conclude that $L$ is at most $\e^{-O(\ddim)}$ times heavier than $L$, which is $\e^{-O(\ddim)}$ times heavier than the minimum spanning tree on $X$.
\end{proof}

\section{Light Spanner with Worst-Case Recourse Bounds}
\label{S:recourse}

We now describe how to maintain our light spanner $L$, satisfying Invariants~\ref{inv:stretch} and \ref{inv:light}. The current section focuses on describing a spanner with a small \emph{recourse} bound, and proving that our updates maintain Invariants~\ref{inv:stretch} and \ref{inv:light}. 
In Section~\ref{S:distance-oracle}, we explain how to perform these updates quickly. 
We maintain a data structure \EMPH{\DS}. It stores:

\begin{itemize}
    \item a set of points \EMPH{$\points$}
    \item a net-tree \EMPH{$\net$} on the points $\points$
    \item an $\e$-net-tree spanner \EMPH{$\sparse$} of $\points$ with respect to $\net$ 
    \item a subgraph $\EMPH{\light} \subseteq \sparse$ called the \EMPH{light spanner}, such that Invariants~\ref{inv:stretch} and \ref{inv:light} hold.
\end{itemize}

All these variables are initialized as $\varnothing$. The data structure supports two operations---insertion and deletion of a point $x$. We say that a procedure $\textsc{Insert}(x)$ \EMPH{correctly implements insertion} of point $x$ if the procedure assigns $\points \gets \points \cup \set{x}$ and updates $\net$, $\sparse$, and $\light$ to satisfy the descriptions above. Similarly, a procedure $\textsc{Delete}(x)$ \EMPH{correctly implement deletion} of a point $x$  if it assigns $\points \gets \points \setminus \set{x}$ and updates $\net$, $\sparse$, and $\light$.

\subsection{Insertion}
To insert a point $x$, we first update the net-tree according to \cite{gao2004deformable} (as summarized in our Lemma~\ref{lem:deformable})
and then recompute the spanner edges $\light$ nearby $x$ to satisfy the invariants. The pseudocode is given below. Our key insight is that when we recompute spanner edges at level $i$, we only need to recompute edges with endpoints in $B(x, O(2^i))$.

\begin{figure}[h!]
\centering
\begin{algorithm}
    \ul{$\texttt{DS}.\textsc{Insert}(x)$}:\+
    \\ $\points \gets \points \cup \set{x}$
    \\ update $\net$ and $\sparse$ according to Lemma~\ref{lem:deformable}
    \\ $\texttt{DS}.\textsc{Recompute}(x)$\-
    \\    
    \\ \ul{$\texttt{DS}.\textsc{Recompute}(x)$:}\+
    \\ for every scale $i \gets 0, 1, \ldots, \log \Phi$:\+
    \\      for every scale-$i$ edge $(u,v)$ of $\sparse$ with $u,v \in B(x, 4 \cdot 2^i)$:\+
    \\          if $\dist^*(u,v) > (1+\e) \cdot \dist(u,v)$:\+
    \\  $\light \gets \light \cup \set{(u,v)}$\-
    \\ else:\+
    \\ $\light \gets \light \setminus \set{(u,v)}$\-
\end{algorithm}
\label{fig:insert-recourse}
\end{figure}

\begin{lemma}
\label{lem:insert-invariants}
The procedure $\DS.\textsc{Insert}(x)$ correctly implements the insertion of a point $x$.
\end{lemma}
\begin{proof}
Let $\cN^{\rm old}$, $S^{\rm old}$, $L^{\rm old}$, and $X^{\rm old}$ denote $\net$, $\sparse$, $\light$, and $\points$ (respectively) before the insertion process.
For brevity, we write $\cN$, $S$, $L$, and $X$ instead of $\net$, $\sparse$, $\light$, and $\points$.
By Lemma~\ref{lem:deformable}, the variables $\points$, $\net$ and $\sparse$ are correctly maintained.
Moreover, we have that $L \subseteq S$, because (by description of the algorithm) $S^{\rm old} \subseteq S$ and (by assumption) $L^{\rm old} \subseteq S^{\rm old}$.

    It remains to show that that every edge $(u,v)$ in $S$ satisfies Invariants~\ref{inv:stretch} and \ref{inv:light}.
    Suppose that edge $(u,v)$ is at scale $i$ (that is, $\delta(u,v)\in (2^{i-1},2^{i}]$).
    Note that the distance $\dist^*(u,v) = \dist_{L[i]}(u,v)$ 
    depends only on the edges in $L$ with scale strictly less than $i$; thus, it suffices to show that $(u,v)$ satisfies the two invariants at the moment after the for-loop processes scale $i$.
    If $u,v \in B(x, 4 \cdot 2^i)$, then $(u,v)$ satisfies the two invariants by description of the algorithm. Now suppose that some endpoint, without loss of generality $u$, is not in $B(x, 4 \cdot 2^i)$. In this case, $(u,v)$ is in $L$ if and only if $(u,v)$ is in $L^{\rm old}$. There are two cases:
    \begin{itemize}
        \item \textbf{\boldmath Suppose that $(u,v)$ is not in $L$.} Thus $(u,v)$ is not in $L^{\rm old}$. To satisfy Invariant~\ref{inv:stretch}, we must show $\dist_{L[i]}(u,v) \le (1+\e) \cdot \dist(u,v)$. In fact, we show that either $\dist_{L[i]}(u,v) \le \dist_{L^{\rm old}[i]}(u,v)$ or $\dist_{L^{\rm old}[i]}(u,v) > 2 \dist(u,v)$; see Figure~\ref{fig:recourse}. This suffices, as $L^{\rm old}$ satisfies Invariant~\ref{inv:stretch}: we conclude that $\dist_{L[i]}(u,v) \le \dist_{L^{\rm old}[i]}(u,v) \le (1+\e) \dist(u,v)$ as desired. Consider a shortest path $P$ in $L^{\rm old}[i]$ between $u$ and $v$. 
        The description of the algorithm implies that the only edges that that change between $L[i]$ and $L^{\rm old}[i]$ are in the ball $B(x, 4 \cdot 2^{i-1})$. Because we assumed $\dist(u,x) > 4 \cdot 2^i$, triangle inequality implies that endpoints of any edge $(u',v')$ in $L^{\rm old}[i] \setminus L[i]$ are far from $u$: specifically, $\dist(u,u') \ge 4 \cdot 2^i - 4 \cdot 2^{i-1} = 2 \cdot 2^i$, and so $\dist_{L^{\rm old}}(u,u') \ge 2 \cdot 2^i$.
        If $P$ has length $\ge 2 \cdot 2^i$ then we are done; otherwise,
        $P$ does not include any edges in $L^{\rm old}[i] \setminus L[i]$, and so $P$ is also a path in $L[i]$ as desired.
    \begin{figure}[h!]
    \centering
    \includegraphics[width=0.4\linewidth]{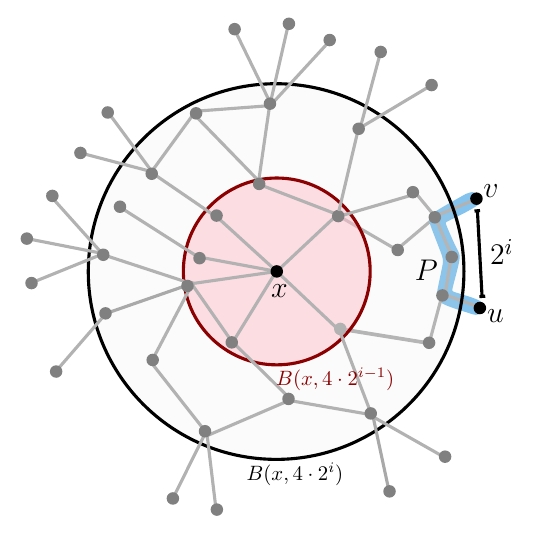}
    \caption{A light spanner $L^{\rm old}$. For a given scale $i$, the spanners $L[i]$ and $L^{\rm old}[i]$ differ only within the ball $B(x, 4 \cdot 2^{i-1})$. A scale-$i$ pair $(u,v)$ outside of $B(x, 4 \cdot 2^i)$. A short path $P$ between $u$ and $v$ in $L^{\rm old}$ does not wander inside $B(x, 4 \cdot 2^{i-1})$, so $P$ is also in $L$.}
    \label{fig:recourse}
    \end{figure}

        \item \textbf{\boldmath Suppose $(u,v)$ is in $L$.} Thus $(u,v)$ is in $L^{\rm old}$. To satisfy Invariant~\ref{inv:light}, we must show that $\dist_{L[i]}(u,v) > (1+ \frac \e 3) \cdot \dist(u,v)$. We claim that either $\dist_{L[i]}(u,v) \ge \dist_{L^{\rm old}[i]}(u,v)$, or $\dist_{L[i]}(u,v) \ge 2 \dist(u,v)$. This suffices to prove the claim, because $L^{\rm old}$ satisfies Invariant~\ref{inv:light} and so $\dist_{L^{\rm old}[i]}(u,v) > (1+ \frac \e 3)\dist(u,v)$; thus, in either case, $\dist_{L[i]}(u,v) > (1+ \frac \e 3) \dist(u,v)$. Let $P$ be a shortest path between $u$ and $v$ in $L[i]$. If $\dist_{L[i]}(u,v) \ge \dist_{L^{\rm old}[i]}(u,v)$ then we are done. Otherwise, if $\dist_{L[i]}(u,v) < \dist_{L^{\rm old}[i]}(u,v)$, then $P$ must contain some edge of $L[i]$ that is not in $L^{\rm old}[i]$. But the description of the algorithm implies that any such edge has endpoints in $B(x, 4 \cdot 2^{i-1})$. By triangle inequality and our assumption that $\dist(u,x) \ge 4 \cdot 2^i$, we have $\dist(u, u') \ge 2 \cdot 2^i$, and so $P$ has length at least $2 \cdot 2^i$ as desired. 
        \aftermath
    \end{itemize}
\end{proof}

We remark that the $\textsc{Insert}$ procedure has bounded recourse: at most $\e^{-O(\ddim)} \log \Phi$ edges change in the spanner $L$. This is because, for every scale $i \in \lceil \log \Phi \rceil$, there are at most $\e^{-O(\ddim)}$ scale-$i$ edges in the $\e$-net-tree spanner in $B(x, 4 \cdot 2^i)$, by definition of net-tree spanner and packing bound (Observation~\ref{obs:packing}). 

\subsection{Deletion}
The procedure for deleting a point $x$ is similar to insertion. The procedure \EMPH{$\DS.\textsc{Delete}(x)$} first updates the net tree in a standard manner by greedily fixing each net \cite{gao2004deformable,KL04}, then deletes all edges in $L$ that were incident to $x$, and finally recompute the spanner edges nearby $x$ using the $\DS.\textsc{Recompute}$ procedure. The pseudocode and analysis are fairly similar to insertion procedure (albeit slightly more complicated).
Like insertion, the deletion procedure has recourse at most $\e^{-O(\ddim)} \cdot \log \Phi$. 
The pseudocode for the $\DS.\textsc{Delete}$ procedure is given in Figure~\ref{fig:delete-recourse}.

\begin{figure}[h!]
\centering
\begin{algorithm}
    \ul{$\DS.\textsc{Delete}(x)$}:\+
    \\ $\points \gets \points \setminus \set{x}$
    \\ update $\net$ and $\sparse$ according to Lemma~\ref{lem:deformable}
    \\ delete every edge in $\light$ incident to $x$
    \\ \DS.\textsc{Recompute}$(x)$
\end{algorithm}
\caption{Deletion procedure with bounded recourse}
\label{fig:delete-recourse}
\end{figure}

\begin{lemma}
\label{lem:delete-invariants}
The procedure $\DS.\textsc{Delete}(x)$ correctly implements the deletion of a point $x$.    
\end{lemma}
\begin{proof}
The proof is very similar to that of Lemma~\ref{lem:insert-invariants}.
    It is immediate from the description of the algorithm and Lemma~\ref{lem:deformable} that that $\net$ remains a valid net tree for $\points$ after the deletion procedure, and $\sparse$ is still an $\e$-net-tree spanner of $\net$. Let $\cN^{\rm old}$, $S^{\rm old}$, and $L^{\rm old}$ denote $\net$, $\sparse$, and $\light$ (respectively) before the deletion process. We write $\cN$, $S$, and $L$ instead of $\net$, $\sparse$, and $\light$.

    First we observe that $L \subseteq S$. Indeed, if $(u,v)$ is an edge added to $L$ by $\DS.\textsc{Recompute}(x)$, then $(u,v)$ is in $S$ and we are done. Otherwise, $(u,v)$ is in $L^{\rm old}$ and thus is in $S^{\rm old}$. In this case we are not immediately done, because we no longer have the property that $S^{\rm old} \subseteq S$; nevertheless, observe that the only edges in $S^{\rm old} \setminus S$ are edges incident to $x$. Moreover, $(u,v)$ is some scale-$i$ edge with one endpoint not in $B(x, 4 \cdot 2^i)$,
    meaning $(u,v)$ is not incident to $x$. Thus $(u,v)$ is in $S$.

    Next we show that every scale-$i$ edge $(u,v)$ in $S$ satisfies Invariants~\ref{inv:stretch} and \ref{inv:light}. If $(u,v)$ has both endpoints in $B(x, 4 \cdot 2^i)$, then it is processed by the call to $\DS.\textsc{Recompute}(x)$ and thus satisfies the invariants. So now suppose that $\dist(u, x) > 4 \cdot 2^i$. There are two cases.
    \begin{itemize}
        \item \textbf{\boldmath Suppose $(u,v)$ is not in $L$.} In this case, the description of the algorithm implies that $(u,v)$ is not in $L^{\rm old}$. This means that either $(u,v)$ is in $S^{\rm old} \setminus L^{\rm old}$ or $(u,v)$ was not an edge in $S^{\rm old}$. We claim (and prove below) that the only edges in $S \setminus S^{\rm old}$ at scale $i$ have both endpoints in $B(x, 4 \cdot 2^i)$. This implies that $(u,v)$ was not an edge in $S^{\rm old}$.
        From here, the proof is exactly the same as in Lemma~\ref{lem:insert-invariants}, and we repeat it here. We show that $\dist_{L[i]}(u,v) \le \dist_{L^{\rm old}[i]}(u,v) \le (1+\e) \dist(u,v)$. Indeed, the shortest path $P$ in $L_i^{\rm old}$ between $u$ and $v$ has length at most $(1+\e) \dist(u,v)$ by Invariant~\ref{inv:stretch}, and this path also appears in $L_i$, because the only edges that change between $L_i^{\rm old}$ and $L_i$ are in $B(x, 4 \cdot 2^{i-1})$ and thus are at least distance $2^i$ away from $u$.

        It remains to show that if $(u,v) \in S \setminus S^{\rm old}$ then $u,v \in B(x, 4 \cdot 2^i)$. Because $(u,v)$ is a scale-$i$ edge of $S$, the endpoints $u$ and $v$ belong to $\e 2^i$-net of the net-tree $\net$. In other words, letting $i' \coloneqq \log(\e 2^i)$, both $u$ and $v$ belong to the net $N_{i'}$. Moreover, because $(u,v)$ is not in $S^{\rm old}$, one endpoint (WLOG $u$) is $N_{i'}^{\rm old}$ but not $N_{i'}$. By Lemma~\ref{lem:net-delete-local}, we have $\dist(u,x) \le 2^{i'}$. By triangle inequality, $\dist(v,x) \le 2^{i'} + 2^i < 2 \cdot 2^i$. We conclude that $u,v \in B(x, 4 \cdot 2^i)$ as desired.
        \item \textbf{\boldmath Suppose $(u,v)$ is in $L$.} This case is exactly the same as in the proof of Lemma~\ref{lem:insert-invariants}, and we repeat it. We show that either $\dist_{L[i]}(u,v) \ge \dist_{L^{\rm old}[i]}(u,v)$ or $\dist_{L[i]}(u,v) > 2 \dist(u,v)$.
        If $\dist_{L^{\rm old}[i]}(u,v) \ge \dist_{L[i]}(u,v)$, we are done because $(u,v)$ is in $L^{\rm old}$ and (by Invariant~\ref{inv:stretch}) we have $\dist_{L^{\rm old}[i]}(u,v) > (1+ \frac \e 3) \dist(u,v)$. Otherwise, because the only edges that change between $L^{\rm old}[i]$ and $L[i]$ lie in $B(x,4 \cdot 2^{i-1})$, we conclude that the shortest path between $u$ and $v$ must have length at least $2 \dist(u,v)$. Either way, $(u,v)$ satisfies Invariant~\ref{inv:stretch}.
    \aftermath
        \end{itemize}
\end{proof}

Finally, we pause to remark that in our proof of Lemma~\ref{lem:insert-invariants} (and also Lemma~\ref{lem:delete-invariants}), we have actually proved the following claim, which will be helpful later.
\begin{lemma}
\label{lem:dist-unchanged}
Consider the data structure $\DS$ before and after running $\DS.\textsc{Insert}(x)$ or $\DS.\textsc{Delete}(x)$. Let $S^{\rm old}$ (resp. $L^{\rm old}$) denote $\sparse$ (resp. $\light$) before the insertion/deletion is performed.
Let $(u,v)$ be a scale-$i$ edge in $\sparse$, with some endpoint outside of $B(x, 4\cdot 2^i)$. Then $(u,v)$ is an edge in $S^{\rm old}$, and either (1)
$\dist_{\light[i]}(u,v) = \dist_{L^{\rm old}[i]}(u,v)$, or (2) both $\dist_{\light[i]}(u,v)$ and $\dist_{L^{\rm old}[i]}(u,v)$ are larger than $2 \dist(u,v)$.
\end{lemma}

\section{Implementing Updates Quickly}
\label{S:distance-oracle}

In this section, we show how to modify the data structure so that the $\DS.\textsc{Insert}$ and $\DS.\textsc{Delete}$ procedures can be executed quickly. By prior work on the net-tree spanner (Lemma~\ref{lem:deformable}), the net-tree $\net$ and the $\e$-net-tree spanner $\sparse$ can be maintained in $O(\log \Phi)$ time. The $\DS.\textsc{Recompute}$ procedure iterates over each of the $\log \Phi$ scales. For each scale $i$, it examines the $\e^{-O(\ddim)} \log \Phi$ many scale-$i$ edges in $B(x,4 \cdot 2^i)$; by Corollary~\ref{cor:expanded-fast-ball} these edges can be found in $\e^{-O(\ddim)} = O(1)$ time. 
Each examined edge $(u,v)$ is added or removed from $\light$ depending on $\dist^*(u,v)$. This means that the bottleneck for a fast update time is the computation of $\dist^*(u,v)$:

\begin{lemma}
\label{lem:basic-runtime}
    The procedure $\DS.\textsc{Insert}$ and $\DS.\textsc{Delete}$ run in $\e^{-O(\ddim)} \cdot\tilde O (\log \Phi)$ time, plus the time needed for $\e^{-O(\ddim)} \cdot \log \Phi$ many computations of $\dist^*(u,v)$.
\end{lemma}

Unfortunately, we don't know how to design a data structure that maintains these distances. To get around this, we first observe that it suffices to maintain an \emph{approximation} for $\dist^*(u,v)$.

\begin{definition}
We say a value $\tilde d$ is an \EMPH{$\alpha$-approximation} for a value $d$ if $d \le \tilde d \le \alpha \cdot d$.
     For any function $d(u,v)$ on pairs of points (we will take $d(u,v) = \dist^*(u,v)$ and $d(u,v) = \dist_{\light}(u,v)$), we say that $\tilde d(u,v)$ is a \EMPH{coarse $\alpha$-approximation} for $d(u,v)$ if:
     \[
\tilde{d}(u,v)=\begin{cases}
\text{if }d(u,v)\le2\delta(u,v)\text{ then } & d(u,v)\le\tilde{d}(u,v)\le\alpha\cdot d(u,v)\\
\text{if }d(u,v)>2\delta(u,v)\text{ then } & \tilde{d}(u,v)\ge2\cdot\delta(u,v)
\end{cases}
\]
\end{definition}
In other words, a coarse $\alpha$-approximation $\tilde d(u,v)$ for $\dist^*(u,v)$ lets us either detect when $\dist^*(u,v)$ is very large or provides an $\alpha$-approximation for $\dist^*(u,v)$. Throughout this section, we set $\alpha \coloneqq 1 + \frac \e 3$. As we show in Lemma~\ref{lem:fast-invariants}, even if we used coarse $\alpha$-approximations $\tilde \dist^*(u,v)$ instead of the real distances $\dist^*(u,v)$ in the $\DS.\textsc{Recompute}$ procedure, our spanner $L$ still satisfies Invariants~\ref{inv:stretch} and \ref{inv:light} after the \DS.\textsc{Insert} and \DS.\textsc{Delete} procedures.

\paragraph{Augmenting the data structure with distance estimates.} We now show that we can modify our data structure to efficiently maintain coarse $\alpha$-approximations of $\dist^*(u,v)$ for every edge $(u,v)$ in the $\e$-net-tree $\sparse$. In fact, we maintain something stronger: we need to store several auxiliary values in order to maintain the coarse approximations.
Let \EMPH{$\kappa$} $\ge \frac{1024}{3}$ be a sufficiently large constant, and let \EMPH{$\e_{\rm small}$} $\gets \frac{\e}{3 \cdot \kappa \cdot \log \Phi}$.  Our data structure $\DS$ maintains (in addition to the points $\points$, the net-tree $\net$, the $\e$-net-tree spanner $\sparse$, and the light spanner $\light$) the following objects:

\begin{itemize}
    \item An  $\e_{\rm small}$-net tree spanner \EMPH{$\DS.S_{\rm small}$} for $\net$
    \item For every scale $i \in [\log \Phi]$ and every edge $(u,v)$ in $\DS.S_{\rm small}$ at scale $i$, a value \EMPH
    {$\DS.\tilde \dist^*(u,v)$} which is a \emph{coarse} $(1 + \kappa\cdot i \cdot \e_{\rm small})$-approximation for $\dist^*(u,v)$.
    \item For every scale $i \in [(\log \Phi) - 2]$ and every edge $(u,v)$ in $\DS.S_{\rm small}$ at scale $i$, a value \EMPH
    {$\DS.\tilde \dist_L(u,v)$} which is a $(1+\kappa\cdot i \cdot \e_{\rm small})$-approximation for $\dist_{\light}(u,v)$.
\end{itemize}

Observe that this data structure does indeed maintain a coarse $\frac \e 3$-approximation for every edge in $\DS.S$: this is because $\DS.S \subset \DS.S_{\rm small}$,
and $\kappa \cdot i \cdot \e_{\rm small} < \frac \e 3$ so the estimates $\DS.\tilde \dist^*(u,v)$ are coarse $\frac \e 3$-approximations for $\dist^*(u,v)$. The reason we maintain estimates for every edge in $\DS.S_{\rm small}$ and not just every edge in $\DS.S$ is to control the distortion, which will accumulate by a $(1+O(\e_{\rm small}))$ factor at each scale. 

We highlight two differences between the estimates $\DS.\tilde \dist^*(u,v)$ and $\DS.\tilde \dist_L(u,v)$, for a scale-$i$ edge $(u,v)$. First, the value $\DS.\tilde \dist^*(u,v)$ seeks to estimate $\dist^*(u,v)$, which is the distance between $u$ and $v$ \ul{in the subgraph $\light[i]$} which consists of all edges in $\light$ with scale strictly smaller than $i$. On the other hand, $\DS.\tilde \dist_L(u,v)$ estimates distances \ul{in the graph $\light$}; note that the shortest path between $u$ and $v$ in $\light$ may include edges in $\light$ at scale $i$ and $i+1$, in addition to edges in $\light[i]$ (see Observation~\ref{obs:small-L}). The second key difference is that $\DS.\tilde \dist^*(u,v)$ is only required to be a $(1+\kappa \cdot i \cdot \e_{\rm small})$-approximation of $\dist^*(u,v)$ when $\dist^*(u,v) \le 2 \dist(u,v)$, whereas $\DS.\tilde \dist_L(u,v)$ is always a $(1+\kappa \cdot i \cdot \e_{\rm small})$-approximation of $\dist_\light(u,v)$.
Why do we need these two types of estimates? Ultimately, our algorithm for \textsc{Insert} and \textsc{Delete} will only use the estimates $\DS.\tilde \dist^*(u,v)$, to determine whether or not the edge $(u,v)$ should be added. But in order to compute $\DS.\tilde \dist^*(u,v)$ for some pair $(u,v)$ at scale $i$, we will need to estimate $\dist_\light(u',v')$ for pairs $(u',v')$ at scale $\le i-3$.

\paragraph{The fast insert/delete procedures.}
Below, we give  pseudocode for a modified $\DS.\textsc{Recompute}$ procedure called \EMPH{$\DS.\textsc{RecomputeFast}$}, which uses the estimates $\DS.\tilde \dist^*(u,v)$ instead of the true distances $\dist^*(u,v)$. We define the procedures \EMPH{$\DS.\textsc{InsertFast}$} and \EMPH{$\DS.\textsc{DeleteFast}$} to be identical to $\DS.\textsc{Insert}$ and $\DS.\textsc{Delete}$, except that they use $\DS.\textsc{RecomputeFast}$ instead of $\DS.\textsc{Recompute}$, and they also maintain a $\e_{\rm small}$-net-tree spanner $\DS.S_{\rm small}$ in addition to the $\e$-net-tree spanner $\sparse$.
\begin{figure}[h!]
\centering
\begin{algorithm}
    \ul{$\texttt{DS}.\textsc{RecomputeFast}(x)$:}\+
    \\ for every scale $i \gets 0, 1, \ldots, \log \phi$:\+
    \\      $\DS.\textsc{UpdateDistEstimates}(x,i)$ \Comment{recomputes estimates $\DS.\tilde \dist^*(\cdot,\cdot)$ and $\DS.\tilde \dist_L(\cdot,\cdot)$.}
    \\      for every scale-$i$ edge $(u,v)$ of $\sparse$ with $u,v \in B(x, 8 \cdot 2^i)$:\+
    \\          $\tilde d \gets \DS.\tilde \dist^*(u,v) $ \Comment{a $\frac \e 3$-good estimate for $\dist^*(u,v)$}
        \\          if $\tilde d > (1+\e) \cdot \dist(u,v)$:\+
    \\  $\light \gets \light \cup \set{(u,v)}$\-
    \\ else:\+
    \\ $\light \gets \light \setminus \set{(u,v)}$\-
\end{algorithm}
\label{fig:recompute-fast}
\end{figure}

In Section~\ref{SS:update-dists}, we describe the procedure \DS.\textsc{UpdateDistEstimates} and prove that it correctly maintains estimates for $\dist^*(u,v)$ and $\dist_L(u,v)$.
For now, we just state lemmas which summarize the guarantees of \DS.\textsc{UpdateDistEstimates} (Lemmas~\ref{lem:dist-fast} and \ref{lem:dist-correct}). Assuming the lemmas hold, we prove the correctness of the fast insert/delete procedures (Lemma~\ref{lem:fast-invariants}).
\begin{lemma}
    \label{lem:dist-fast}
    $\DS.\textsc{UpdateDistEstimates}$ runs in time $(\e^{-1} \cdot  \log \Phi)^{O(\ddim)}$.
\end{lemma}
\begin{restatable}{lemma}{distCorrect}
\label{lem:dist-correct}
    Consider the execution of $\DS.\textsc{InsertFast}(x)$ or $\DS.\textsc{DeleteFast}(x)$.
    Let $i\in [\log \Phi]$ be a scale, and consider the data structure $\DS$ at the instant after the execution of the call $\DS.\textsc{UpdateDistEstimates}(x,i)$ during the insertion/deletion process. Assume that every edge $(u,v)$ of $\sparse$ with scale strictly smaller than $i$ satisfies Invariants~\ref{inv:stretch} and \ref{inv:light}. Then:
    \begin{itemize}
        \item for every scale-$i$ edge $(u,v)$ in the $\e_{\rm small}$-net-tree spanner $\DS.S_{\rm small}$,
        the value $\DS.\tilde \dist^*(u,v)$ is a
        coarse $(1 + \kappa \cdot i \cdot \e_{\rm small})$-approximation for $\dist^*(u,v)$.
        \item for every scale-$(i-2)$ edge $(u,v)$ in the $\e_{\rm small}$-net-tree spanner $\DS.S_{\rm small}$, the value $\DS.\tilde \dist_L(u,v)$ is a $(1 + \kappa \cdot i \cdot \e_{\rm small})$-approximation $\dist_\light(u,v)$.
    \end{itemize}
\end{restatable}

\begin{lemma}
\label{lem:fast-invariants}
    The procedure $\DS.\textsc{InsertFast}(x)$ (resp. $\DS.\textsc{DeleteFast}(x)$) correctly implements insertion (resp. deletion) of $x$. It runs in time $(\e^{-1}\cdot \log \Phi)^{O(\ddim)}$.
\end{lemma}
\begin{proof}
    The runtime bound follows from Lemmas~\ref{lem:basic-runtime} and \ref{lem:dist-fast}. The proof of correctness is almost identical to that of Lemmas \ref{lem:insert-invariants} and \ref{lem:delete-invariants}
    except in one place: we need to argue that every scale-$i$ edge $(u,v)$ in $\sparse$ with both endpoints in $B(x, 4 \cdot 2^i)$ satisfies Invariants~\ref{inv:stretch} and \ref{inv:light}. 
    
    For every such edge $(u,v)$, the algorithm checks if $\DS.\tilde \dist^*(u,v) > (1+\e) \dist(u,v)$ and includes $(u,v)$ in $L$ iff the inequality holds. At the instant the algorithm makes the check, we may assume (by induction) that every edge $(u',v')$ in $\sparse$ with scale strictly smaller than $i$ satisfies Invariants~\ref{inv:stretch} and \ref{inv:light}. Thus, by Lemma~\ref{lem:dist-correct}, the value of $\DS.\tilde \dist(u,v)$ is a coarse $(1 + \kappa \cdot i \cdot \e_{\rm small})$-approximation for $\dist^*(u,v)$. As $i \in [\log \Phi]$, in particular we conclude that $\DS.\tilde \dist(u,v)$  is a coarse $(1+\frac \e 3)$-approximation for $\dist^*(u,v)$.
    There are two cases.

    \medskip \noindent \textbf{\boldmath Case 1: $\DS.\tilde \dist (u,v) \le (1+\e) \dist(u,v)$.} In this case, $(u,v)$ is not in $\light$. We claim that $\dist^*(u,v) \le (1+\e) \dist(u,v)$, so that $(u,v)$ satisfies Invariant~\ref{inv:stretch}. Indeed, we have $\dist^*(u,v) < 2 \dist(u,v)$; otherwise, the definition of coarse $\alpha$-approximation means that $\DS.\tilde \dist(u,v) > 2\dist(u,v)$, which contradicts our assumption that $\DS.\tilde \dist(u,v) \le (1+\e) \dist(u,v)$. Thus, the definition of coarse $\alpha$-approximation implies $\dist^*(u,v) \le \DS.\tilde \dist(u,v)$ and so $\dist^*(u,v) \le (1+\e) \dist(u,v)$ as desired.

    \medskip \noindent \textbf{\boldmath Case 2: $\DS.\tilde \dist (u,v) > (1+\e) \dist(u,v)$.} In this case, $(u,v)$ is added to $\light$. We claim that $\dist^*(u,v) > (1+\frac{\e}{3}) \dist(u,v)$, so that $(u,v)$ satisfies Invariant~\ref{inv:light}. For contradiction, suppose that $\dist^*(u,v) \le (1+\frac{\e}{3}) \dist(u,v)$. The definition of coarse $(1+\frac \e 3)$-approximation implies $\DS.\tilde \dist(u,v) \le (1+ \frac \e 3) \cdot \dist^*(u,v)$. Combining these two inequalities, we conclude that
    \(\DS.\tilde \dist(u,v) \le \left(1+\frac \e 3 \right) \left(1+\frac \e 3\right) \dist(u,v) \le (1+\e) \dist(u,v),\)
    contradicting the Case 2 assumption.
\end{proof}
Theorem~\ref{thm:main} follows from Lemmas~\ref{lem:fast-invariants},\ref{lem:inv-implies-stretch}, and \ref{lem:inv-implies-light} (after rescaling $\e \gets \e / 3$).

\subsection{The procedure \textsc{UpdateDistEstimates}}
\label{SS:update-dists}
It remains to describe the $\textsc{UpdateDistEstimates}(x,i)$ procedure and prove Lemma~\ref{lem:dist-correct}. This procedure recomputes the distance estimates $\DS.\tilde \dist^*(\cdot, \cdot)$ for the scale-$i$ edges of $\DS.S_{\rm small}$.
It is described in pseudocode below, and summarized in text immediately after.

\begin{figure}[h!]
\centering
\begin{algorithm}
    \ul{$\texttt{DS}.\textsc{UpdateDistEstimates}(x,i)$:}\+
    \\      for every scale-$(i-2)$ edge $(u,v)$ of $\DS.S_{\rm small}$ with $u,v \in B(x, 4 \cdot 2^i)$:\+
    \\  $\DS.\tilde \dist_L(u,v) \gets \DS.\textsc{Estimate}(u,v,i)$\-
    \\      for every scale-$i$ edge $(u,v)$ of $\DS.S_{\rm small}$ with $u,v \in B(x, 4 \cdot 2^i)$:\+
    \\  $\DS.\tilde \dist^*(u,v) \gets \DS.\textsc{Estimate}(u,v,i)$\-\-
    \\
    \\ \ul{$\texttt{DS}.\textsc{Estimate}(u,v,i)$:}\+
    \\      \Comment{Returns an estimate of the distance $\dist_{\light [i]}(u,v)$}
    \\      $i' \gets \max \set{0, \log (\e_{\rm small} 2^i)}$
    \\      $H \gets$ graph with vertex set $N_{i'} \cap B(x, 7 \cdot 2^i)$
    \\      for every pair $(u',v')$ in $V(H)$ with $2^{i - 3} \le \dist(u',v') < 2^{i-1}$:\+
    \\          \Comment{Add all edges with scale $i -1$ or $i - 2$}
    \\          if $(u',v') \in \light$, then add edge $(u',v')$ to $H$ with weight $\dist(u',v')$\-
    \\      for every pair $(u',v')$ in $V(H)$ with $\dist(u',v') < 2^{i - 4}$:\+
    \\      \Comment{Estimate distances using scale $\le i -3$ edges}
    \\          add edge $(u',v')$ to $H$ with weight $\DS.\tilde \dist_L(u',v')$\-
    \\      return $\dist_H(u,v)$

\end{algorithm}
\label{fig:update-dists}
\end{figure}

The key subroutine is a procedure $\textsc{Estimate}(u,v,i)$ which estimates $\dist_{\light [i]}(u,v)$---that is, distances using edges of $\light$ with scale strictly smaller than $i$---under the assumption that the data structures stores accurate estimates $\DS.\tilde \dist_{L}(u',v')$ for $\dist_{\light [i]}$ for every pair of $(u',v')$ at scale $\le i - 3$.
To do this, the estimation procedure constructs a constant-size \EMPH{sketch graph} $H$ and returns the distance between $u$ and $v$ in $H$. The sketch graph $H$ contains every scale-$(i-1)$ and scale-$(i-2)$ edge in $\light$ that is close to $u$ or $v$. Additionally, (roughly speaking) $H$ contains edges $(u',v')$ between the endpoints of different large-scale edges\footnote{in practice, we also add edges between some net points which are not endpoints of a large-scale edge}, which are weighted using the estimates $\DS.\tilde \dist_L(u',v')$, for pairs $(u',v')$ at scale $\le i - 3$. These two types of edges are sufficient to preserve the shortest path between $u$ and $v$ in $\light [i]$.
We first observe that $\DS.\textsc{UpdateDistEstimates}$ runs in $(\e^{-1} \cdot \log \Phi)^{O(\ddim)}$ time.
\begin{proof}[of Lemma~\ref{lem:dist-fast}]
    Each call to $\DS.\textsc{Estimate}$ runs in time $(\e^{-1} \cdot \log \Phi)^{O(\ddim)}$: indeed, it follows from the packing bound (Observation~\ref{obs:packing}) that $H$ has at most $(\e^{-1} \cdot\log \Phi)^{O(\ddim)}$ vertices, and it follows from Corollary~\ref{cor:expanded-fast-ball} that these vertices can be found in $\e^{-O(\ddim)}$ time. Overall, $\DS.\textsc{UpdateDistEstimates}$ makes $\e^{-O(\ddim)} \cdot \log \Phi$ calls to $\DS.\textsc{Estimate}$.
\end{proof}

We now prove Lemma~\ref{lem:dist-correct}, that is, correctness of the algorithm.
%
Our proof is by induction on $i$. Given a scale $i$, we (inductively) assume that Lemma~\ref{lem:dist-correct} holds for all scales $\le i-1$; in particular, we assume that for every edge $(u',v')$ of $\sparse_{\rm small}$ with scale $\le i -3$, the value $\DS.\tilde \dist_L(u',v')$ is a $(1 + \kappa \cdot(i-3) \cdot\e_{\rm small})$-approximation for $\dist_L(u'v')$. We also assume that Invariants~\ref{inv:stretch} and \ref{inv:light} hold for all edges of $\sparse$ with scale $\le i-1$.
Throughout this section, we write $\cN$, $S$, $S_{\rm small}$, and $L$ instead of $\net$, $\sparse$, $\sparse_{\rm small}$ and $\light$ (after the insertion/deletion procedure terminates).
Our key claim is that $\DS.\textsc{Estimate}(u,v,i)$ returns a coarse approximation of $\dist_{L[i]}(u,v)$. We begin with a simple observation that $\dist_{L}(u',v') = \dist_{L[i]}(u',v')$, for pairs of vertices $(u',v')$ at scale $\le i - 2$.

\begin{observation}
\label{obs:small-L}
    Let $a$ and $b$ be two vertices with $\dist(a,b) \le 2^j$ for some scale $j \le i - 2$. 
    We have that $\dist_L(a,b) < 2 \cdot \dist(a,b)$. We conclude that that the shortest path between $a$ and $b$ in $L$ uses edge with scale at most $j+1$; that is, $\dist_{L}(a,b) = \dist_{L[j+2]}(u',v')$.
\end{observation}
\begin{proof}
    By assumption, our spanner $L$ satisfies Invariant~\ref{inv:stretch} for every edge in $S$ with scale $\le i - 1$. We would like to apply Lemma~\ref{lem:inv-implies-stretch} to conclude that $\dist_L(a,b) \le (1+3\e) \dist(a,b) < 2 \dist(a,b)$. However, we cannot do this directly, because Lemma~\ref{lem:inv-implies-stretch} assume that Invariant~\ref{inv:stretch} holds for \emph{all} edges in $S$, not just those with scale $\le i -1$. This issue is easy to fix: simply consider the graph $L^+$ which is the union of $L$ and all scale-$j$ edges in $S$ with $j \ge i$. Clearly $L^+$ satisfies Lemma~\ref{lem:inv-implies-stretch}, so it is a $(1+\e)$-spanner. This means that the shortest path between $a$ and $b$ in $L^+$ only uses edges with scale strictly smaller than $j + 2 \le i$: an edge with scale $\ge i$ has length at least $2^{i-1}$, but we have $\dist_L(u,v)$ is strictly smaller than $2 \dist(u,v) < 2 \cdot 2^{j} \le 2^{i-1}$. We conclude that this path also appears in $L$.
\end{proof}

\begin{lemma}
\label{lem:estimate-correct}
Let $(u,v)$ be an edge of $S_{\rm small}$ with scale $\le i$, and $u,v \in B(x, 4 \cdot 2^i)$. Then $\DS.\textsc{Estimate}(u,v,i)$ returns a coarse $(1 + \kappa \cdot i \cdot \e_{\rm small})$-approximation for $\dist_{L[i]}(u,v)$.
\end{lemma}
\begin{proof}
    Let $H$ be the sketch graph constructed by $\DS.\textsc{Estimate}(u,v,i)$.
    
    \paragraph{Lower bound.} First we prove $\dist_{L[i]}(u,v) \le \dist_H(u,v)$. It suffices to show that, for any edge $(u',v')$ in $H$ with weight $w$, we have $\dist_{L[i]}(u,v) \le w$. There are two types of edges that get added to $H$. In the first case, $(u',v')$ is an edge of $\light$ with scale either $i-2$ or $i-1$, and it belongs to $\light$; in this case, we have $w = \dist(u',v') = \dist_{L[i]}(u,v)$.
    In the second case, $(u',v')$ has scale at most $i-3$, and $w = \DS.\tilde \dist_L(u',v')$. By Observation~\ref{obs:small-L} and induction hypothesis, we have $\dist_{L[i]}(u,v) = \dist_{\light}(u,v) \le \DS.\tilde \dist_L(u',v') = w$ as desired.

    \begin{figure}[h!]
    \centering
    \includegraphics[width=0.97\linewidth]{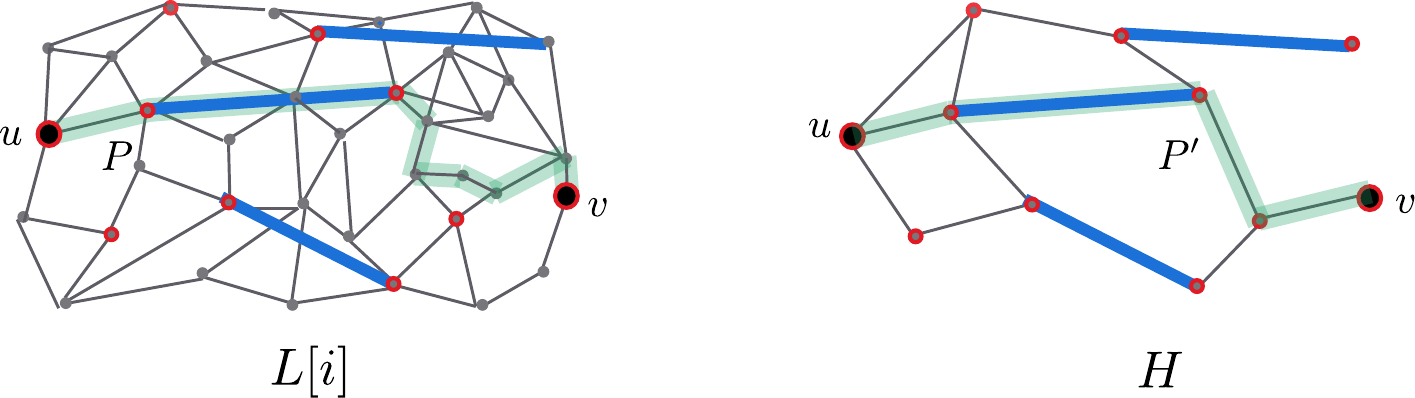}
    \caption{A depiction of the upper bound proof of Lemma~\ref{lem:estimate-correct}. The graph $L[i]$, and the sketch graph $H$. The net points $N_{i'}$ are outlined in red, in both $L[i]$ and $H$. The ``large-scale'' edges of $L[i]$, with scale $(i-1)$ and $(i-2)$, are drawn in blue. A pair $(u,v)$ and the shortest path $P$ between them in $L[i]$ is marked in green, and the corresponding path $P'$ in $H$ is marked in green.}
    \label{fig:distance-oracle-example}
\end{figure}

    \paragraph{Upper bound.} We now show that $\dist_H(u,v) \le (1+(i \cdot \e_{\rm small})$ if $\dist_{L[i]}(u,v) \le 2 \dist(u,v)$.\footnote{If $\dist_{L[i]}(u,v) > 2 \dist(u,v)$, then the previous paragraph implies $\dist_H(u,v) \ge 2 \dist(u,v)$, so we are
    done.
    }
    Refer to Figure~\ref{fig:distance-oracle-example}. Consider the shortest path $P$ between $u$ and $v$ in $L[i]$. We argue that we can pull back $P$ onto $H$ with only small distortion. As $\dist_{L[i]}(u,v) \le 2 \dist(u,v)$, path $P$ only contains edges with both endpoints in $B(u, 2 \cdot 2^i)$. By triangle inequality, path $P$ only contains edges with both endpoints in $B(x, 6 \cdot 2^i)$. 
We divide $P$ into \EMPH{subpaths} as follows. Treat $P$ as a directed path starting at $u$ and ending at $v$. Let $\EMPH{$x_0$} \gets u$. For every $j > 0$, we define \EMPH{$x_j$} to be the farthest point along $P$ such that the subpath $P[x_{j-1}:x_j]$ has length smaller than $2^i / 16$; if no such $x_j$ exists (because the edge following $x_{j-1}$ has length at least $2^i / 16$) we set $x_j$ to be the vertex that comes immediately after $x_{j-1}$ in $P$.  We define the subpath $\EMPH{$P_j$} \coloneqq P[x_{j-1}:x_j]$
We write $P$ as the concatenation of subpaths
\[P = P_1 \circ P_2 \circ \ldots \circ P_k.\]
Observe that every subpath $P_j$ either consists of a single edge of $L[i]$ with scale $\ge i - 3$, or it is a path in $L[i-3]$ (ie, it consists of edges of $L$ with scale $\le i-4$). Moreover, there are at most $k = 64$ subpaths in this decomposition, as (by construction) any two consecutive subpaths together have length $\ge 2^i / 16$, and $P$ has length at most $2 \cdot \dist(u,v) \le 2 \cdot 2^i$.

We now construct a path $P'$ in $H$ between $u$ and $v$, using path $P$. Intuitively, we will argue that each subpath $P_j$ can be pulled back to an edge in $H$; this forms the path $P'$. Let $i' = \max(0, \log(\e_{\rm small}\cdot 2^i))$, and let $N_{i'} \in \net$ be the net at scale $i'$. Observe $N_{i'} \cap B(x, 7\cdot 2^i)$ are the vertices of $H$. For every point $x_j$, we define the point \EMPH{$x_j'$} to be the closest point in $N_{i'}$. By definition of net, we have $\dist(x_j,x_j') \le \e_{\rm small} 2^i < 2^i$, and so triangle inequality implies that $x_j' \in B(x, 7 \cdot 2^i)$. That is, each $x_j'$ is a vertex in the sketch graph $H$. Also note that $x_0 = x_0' = u$ and $x_k = x_k' = v$: indeed, $u$ and $v$ are in $N_{i'}$, because $(u,v)$ is a scale-$i$ edge of $S_{\rm small}$.  Finally, we define \EMPH{$P_j'$} to be the shortest path in $H$ between $x_{j-1}'$ and $x_j'$. Define
\[\EMPH{$P'$} = P_1' \circ P_2' \circ \ldots \circ P_k'.\]
Observe that $P'$ is a path in $H$ between $u$ and $v$. We prove below that for every subpath $P_j'$,
\begin{equation}
\label{eq:subpath-dist}
    \dist_H(x_{j-1}',x_j') \le (1+\kappa\cdot(i-3)\cdot\e_{\rm small})\cdot \len(P_j) + 8 \cdot \e_{\rm small} \cdot 2^i.
\end{equation}
Summed over the 64 subpaths of $P$, Equation~\ref{eq:subpath-dist} implies that there is a path in $H$ between $u$ and $v$ with length at most 
\((1+\kappa \cdot (i-3) \cdot \e_{\rm small})\dist_{L[i]}(u,v) + 512 \cdot \e_{\rm small} 2^i.\)
As $\dist_{L[i]}(u,v) \ge \dist(u,v) \ge 2^i/2$ and $\kappa \ge \frac {1024} 3$ is sufficiently large, we conclude \(\dist_H(u,v) \le (1+\kappa \cdot i \cdot \e_{\rm small}) \dist_{L[i]}(u,v)\).

\medskip \noindent \textbf{Proof of Equation~\ref{eq:subpath-dist}.} There are two cases. In the first case, suppose that $P_j$ consists of a single edge in $L[i]$ with scale $(i-1)$ or $(i-2).$ In this case, the endpoints $x_{j-1}$ and $x_j$ of $P_j$ are in the net $N_{i'}$, and so $x_{j-1} = x_{j-1}'$ and $x_{j} = x_{j}'$. 
The graph $H$ includes the edge $(x_{j-1}',x_j')$ with weight equal to $\len(P_j)$, as desired.

In the second case, suppose that that $P_j$ is either a scale-$(i-3)$ edge \emph{or} a path in $L[i-3]$. 
First we show that $(x_{j-1}', x_j')$ is an edge in the $\e_{\rm small}$-net-tree spanner $S_{\rm small}$ with scale $\le i-3$.
As $x_{j-1}',x_j' \in N_{i'}$, the definition of $S_{\rm small}$ implies that it suffices to show that $\dist(x_{j-1}', x_j') < 2^{i-3}$.
To this end, note that if $P_j$ is a scale-$(i-3)$ edge we have $x_{j-1}' = x_{j-1}$ and $x_{j}'=x_j$, and so $\dist(x_{j-1}', x_{j}') = \dist(x_{j-1}, x_j) < 2^{i-3}$. Otherwise, suppose $P_j$ is a path in $P[i-3]$. Observe that the definition of $N_{i'}$ implies $\dist(x_{j},x_{j}') \le \e_{\rm small} \cdot 2^i$, and similarly $\dist(x_{j-1},x_{j-1}') \le \e_{\rm small} \cdot 2^i$. Triangle inequality implies
\[\dist(x_{j-1}', x_j') \le \dist(x_{j-1},x_j) + 2 \e_{\rm small}\cdot 2^i < \frac{2^i}{16} + 2 \e_{\rm small}\cdot 2^i < 2^{i-3}.\]
Thus, the graph $H$ contains an edge $(x_{j-1}', x_j')$ with weight $\DS.\tilde \dist_L(x_{j-1}',x_j')$. By induction hypothesis and Observation~\ref{obs:small-L}, the value $\DS.\tilde \dist_L(x_{j-1}',x_j')$ is a $(1+\kappa\cdot(i-3)\cdot\e_{\rm small})$-approximation for $\dist_{L[i-3]}(x_{j-1}', x_j')$.
Now we claim that
\begin{equation}
\label{eq:subpath-len}
    \dist_{L[i-3]}(x_{j-1}', x_j') \le \len(P_j) + 4 \e_{\rm small} \cdot 2^i.
\end{equation}
Indeed,
$\dist(x_j, x_j') < 2^{i-5}$, so Observation~\ref{obs:small-L} implies we have $\dist_{L[i-3]}(x_{j},x_{j}') \le 2 \e_{\rm small}\cdot 2^i$. Similarly $\dist_{L[i-3]}(x_{j-1},x_{j-1}') \le 2 \e_{\rm small} \cdot 2^i$. Triangle inequality proves Equation~\ref{eq:subpath-len}.
Finally, we conclude
\(
  \dist_H(x'_{j-1},x'_j) \le (1+\kappa \cdot (i-3) \cdot \e_{\rm small})\cdot(\len(P_j + 4 \e_{\rm small} \cdot 2^i) 
  <(1+\kappa \cdot (i-3) \cdot \e_{\rm small})\cdot \len(P_j) + 8\e_{\rm small} \cdot 2^i.
  \)
This proves Equation \ref{eq:subpath-dist}, and completes the proof of Lemma~\ref{lem:estimate-correct}.
\end{proof}
We are now ready to prove Lemma~\ref{lem:dist-correct}.
\begin{proof}[Proof of Lemma~\ref{lem:dist-correct}]
Let $S_{\rm small}^{\rm old}$ and $L^{\rm old}$ denote $\sparse_{\rm small}$ and $\light$ (respectively) before the insertion/deletion process.
    Let $(u,v)$ be a scale-$i$ edge of $S_{\rm small}$. We claim that value $\DS.\tilde \dist^*(u,v)$ is a coarse $(1 + \kappa \cdot i \cdot \e_{\rm small})$-approximation of $\dist^*(u,v)$. Indeed, if $u,v \in B(x, 4 \cdot 2^i)$ then Lemma~\ref{lem:estimate-correct} implies that the data structure stores a coarse approximation. On the other hand, if one endpoint is not in $B(x, 4 \cdot 2^i)$, then Lemma~\ref{lem:dist-unchanged} implies\footnote{Note that Lemma~\ref{lem:dist-unchanged} is stated with  $\DS.\textsc{Insert}$ and $\DS.\textsc{Delete}$ instead of $\DS.\textsc{InsertFast}$ and $\DS.\textsc{DeleteFast}$, and it refers to edges of $\sparse$ instead of $\sparse_{\rm small}$. But the statement still holds even with these changes.} that $(u,v)$ is an edge in $S^{\rm old}_{\rm small}$, and that either $\dist^*(u,v) = \dist_{L^{\rm old}[i]}(u,v)$ or both $\dist^*(u,v)$ and $\dist_{L^{\rm old}[i]}(u,v)$ are larger than $2 \dist(u,v)$; thus, any coarse $\alpha$-approximation for $\dist_{L^{\rm old}[i]}(u,v)$ is a coarse $\alpha$-approximation for $\dist_{L[i]}(u,v) = \dist^*(u,v)$. By assumption, before insertion/deletion we had stored a value $\DS.\tilde \dist^*(u,v)$ which is a $(1 + \kappa \cdot i \cdot \e_{\rm small})$-approximation for $\dist_{L^{\rm}[i]}(u,v)$, so we are done.

    Let $(u,v)$ be a scale-$(i-2)$ edge of $S_{\rm small}$. By an identical argument, $\DS.\tilde \dist_L(u,v)$ is a coarse $(1 + \kappa \cdot i \cdot \e_{\rm small})$-approximation of $\dist_L(u,v)$. By Observation~\ref{obs:small-L}, we have $\dist_L(u,v) < 2 \dist(u,v)$, and so $\DS.\tilde \dist_L(u,v)$ is a (non-coarse) $(1 + \kappa \cdot i \cdot \e_{\rm small})$-approximation of $\dist_L(u,v)$.
\end{proof}

\begin{remark}[Coping with an unknown aspect ratio.]
\label{rem:aspect-ratio}
    While we have written our data structure as though the aspect ratio $\Phi$ is known in advance, it is possible to remove this assumption by following \cite{gao2004deformable}. Rather than explicitly maintaining the set of points in each net $N_i$ for every scale $i \in [\log \Phi]$, \cite{gao2004deformable} implicitly maintain a net for every integer scale $i \in \Z$. 
    Specifically, they maintain a value $i_{\max}$ (such that for all $i \ge i_{\max}$, the net $N_i$ consists of a single point) and a value $i_{\min}$ (such that for all $i \le i_{\min}$, 
    the net $N_i$ consists of the entire point set $X$) with $i_{\max} - i_{\rm min} = O(\log \Phi)$; for scales $i_{\min} \le i \le i_{\max}$, for each point $x$ in $N_i$, they store a list of points in $N_i$ within distance $O(2^i)$ of $x$. \cite{gao2004deformable} show that this implicit net tree can be maintained efficiently.
    In our algorithms for the dynamic spanner, we could remove the assumption that we know $\Phi$ by using the implicit net tree, and replacing the for-loop in $\DS.\textsc{Recompute}$ with a loop from $i\gets i_{\min}$ up to $i_{\max}$. We implicitly maintain values $\DS.\tilde \dist^*(u,v)$ and $\DS.\tilde \dist_L(u,v)$ for every edge $(u,v)$ and every scale $i$ (for scales $i > i_{\max}$ there is only a single point in the net so this is trivial, and for scale $i < i_{\min}$ there are no scale-$i$ edges).
\end{remark}

\section{Conclusion and Open Questions}
We have designed the first dynamic light spanner for point sets in Euclidean and doubling metrics, with update time $(\log \Phi)^{O(\ddim)}$. Our work leaves open two interesting open questions. First, is it possible to maintain a dynamic light spanner with recourse or time bounds \emph{independent} of the aspect ratio $\Phi$ and instead depends only on the number of points $n = |X|$? Second, can we obtain a runtime bound where the exponent of $\log \Phi$ (or $\log n$) is independent of $\ddim$? Ideally, we would like to handle updates in time $\e^{-O(\ddim)} \cdot \log n$, which would match the best runtime for dynamic sparse spanners.

{
\small
\bibliographystyle{alpha}
\bibliography{main}
}

 \appendix   

\section{LSO-Based Spanners Are Not Light}\label{sec:LSO}
Roughly speaking, given a metric space $(X,\delta)$, an $(\tau,\varepsilon)$-LSO is a collection $\Sigma$ of $\tau$ orderings over $X$ such that every pair of points $x,y\in X$ is almost consequence in one of the orderings (see Definition~\ref{def:LSO}).
Chan \etal~\cite{chan2020locality} constructed $(\tilde{O}(\eps)^{-d},\eps)$-LSO for the Euclidean space $\R^d$ (see also \cite{GaoH24}).
Later, Filtser and Le \cite{filtser2022locality} constructed  $(\eps^{-O(d)},\eps)$-LSO for metrics with doubling dimension $d$.
Given a $(\tau,\eps)$-LSO, one can construct a spanner by connecting any two points which are consecutive in some ordering \cite{chan2020locality}. The result is a $(1+O(\eps))$-spanner with $(n-1)\cdot \tau$ edges. 
Chan~\etal~showed that it is possible to efficiently maintain the LSO for $\R^d$, and thus it become very easy to dynamically maintain a sparse spanner.
Later,  Le and La~\cite{la2025dynamic} constructed a dynamic version of the LSO for doubling metrics, thus allowing us to easily maintain dynamic spanner.
Unfortunately, the resulting spanner is by no means light. Indeed, in Lemma~\ref{lem:LSO}  we show that it has logarithmic lightness (even for $\R^1$).

We begin with a formal definition of an LSO (even though the proof of Lemma~\ref{lem:LSO} will go though also using slightly weaker definition presented above).
\begin{definition}[$(\tau,\rho)$-LSO]\label{def:LSO}
    Given a metric space $(X,d_{X})$, we say that a collection $\Sigma$
	of orderings is a $(\tau,\rho)$-LSO if
	$\left|\Sigma\right|\le\tau$, and for every $x,y\in X$, there is
	a linear ordering $\sigma\in\Sigma$ such that (w.l.o.g.\footnote{That is either  $x\preceq_{\sigma}y$ or  $y\preceq_{\sigma}x$, and the guarantee holds w.r.t. all the points between $x$ and $y$ in the order $\sigma$.\label{foot:wlogOrder}}) $x\preceq_{\sigma}y$ and the points between $x$ and $y$ w.r.t. $\sigma$ could be partitioned
	into two consecutive intervals $I_{x},I_{y}$ where $I_{x}\subseteq B_{X}(x,\rho\cdot d_{X}(x,y))$	and $I_{y}\subseteq B_{X}(y,\rho\cdot d_{X}(x,y))$. $\rho$ is called the \emph{stretch} parameter. 
\end{definition}
In Lemma~\ref{lem:LSO} below we prove that LSO-based spanner has logarithmic lightness even for the path graph (which is Euclidean, and has constant doubling dimension).
\begin{lemma}\label{lem:LSO}
Consider a $(\tau,\eps)$-LSO for the set of points $\left(v_1=1,v_2=2,\dots,v_{2^\Phi}=2^\Phi\right)$, 
for some $\eps\in(0,\frac{1}{6})$. Let $H_{{\rm LSO}}$ be the $(1+O(\eps))$-spanner
created by connecting every two points that are consecutive in some
ordering. Then $w(H_{{\rm LSO}})=\Omega(\Phi\cdot2^{\Phi})$. In particular,
$H_{{\rm LSO}}$ has logarithmic lightness.
\end{lemma}
\begin{proof}
For $i\in[0,\Phi]$, let $N_{i}=\left\{ v_{j}\mid j\,\mod\,2^{i}=0\right\} $
be a hierarchical net. For any pair of consecutive net points in $N_{i}$,
$x=v_{q\cdot2^{i}}$ and $y=v_{(q+1)\cdot2^{i}}$ there is some ordering
$\sigma$ such that $x$ and $y$ are almost consecutive in $\sigma$.
In particular, there are points $x',y'$ which are consecutive along
$\sigma$ such that $\delta(x,x')\le\eps\cdot\delta(x,y)=\eps\cdot2^{i}$,
and $\delta(y,y')\le\eps\cdot2^{i}$. Denote this edge by $e_{i,q}$.
$e_{i,q}$ is part of the spanner resulting from the LSO, and $w(e_{i,q})\in\left[1-2\eps,1+2\eps\right]\cdot2^{i}$. 

We argue that for every $(i,q)\ne(i',q')$ it holds that $e_{i,q}\ne e_{i',q'}$.
For the sake of contradiction, suppose otherwise. We continue by case
analysis:
\begin{itemize}
\item $i\ne i'$. Suppose w.l.o.g. that $i<i'$. Then $w(e_{i,q})\le(1+2\eps)\cdot2^{i}<(1-2\eps)\cdot2^{i'}\le w(e_{i',q'})$,
a contradiction.
\item $i=i'$. Suppose w.l.o.g. that $q<q'$. The edge $e_{i,q}$ has its
left endpoint at distance at most $\eps\cdot2^{i}$ from $v_{q\cdot2^{i}}$.
The edge $e_{i',q'}$ has its right endpoint at distance at most $\eps\cdot2^{i}$
from $v_{(q'+1)\cdot2^{i}}$. It follows that 
\begin{align*}
w(e_{i,q}) & \ge\left|(q'+1)\cdot2^{i}-q\cdot2^{i}\right|-2\cdot\eps\cdot2^{i}\\
 & \ge2\cdot2^{i}-2\cdot\eps\cdot2^{i}=(1-\eps)\cdot2^{i+1}>(1+2\eps)\cdot2^{i}\,,
\end{align*}
a contradiction.\aftermath
\end{itemize}
\end{proof}

\end{document}